\documentclass[11pt]{article}
\usepackage[colorlinks, citecolor=blue, linkcolor=black]{hyperref}
\usepackage{fullpage,amssymb,amsmath,amsthm}

\usepackage[linesnumbered,boxed,noend]{algorithm2e}

\newtheorem{THEOREM}{Theorem}

\newtheorem{theorem}{Theorem}[section]

\newtheorem{corollary}[theorem]{Corollary}

\newtheorem{definition}[theorem]{Definition}

\newtheorem{lemma}[theorem]{Lemma}

\newtheorem{observation}[theorem]{Observation}

\newcommand{\N}{\mathbb{N}}

\newcommand{\eqdef}{\stackrel{\Delta}{=}}

\newcommand{\ignore}[1]{}

\begin{document}

\title{\vspace{-2cm}
One (more) line on the most Ancient Algorithm in History}

\author{Bruno Grenet\thanks{LIRMM, Universit\'{e} de Montpellier, CNRS, Montpellier, France. 
Email: {\tt bruno.grenet@lirmm.fr}. }
\and
Ilya Volkovich\thanks{Department of EECS, CSE Division, University of Michigan, Ann Arbor, MI.
Email: {\tt ilyavol@umich.edu}. }}

\date{}

\maketitle

\vspace{-0.3cm}

\begin{abstract}
We give a new simple and short (``one-line'') analysis for the runtime of the well-known Euclidean Algorithm.  While very short simple, the obtained upper bound in near-optimal.
\end{abstract}

\section{Introduction}

Perhaps the most ancient algorithm recorded in history is the Euclidean Algorithm. Published in 300 BC, the algorithm computes the greatest common divisor (gcd) of two integer numbers.
(See Algorithm \ref{alg:Euclid} for the description of the algorithm.) The following is a quote from \cite{Knuth98}:
``[The Euclidean algorithm] is the granddaddy of all algorithms, because it is the oldest nontrivial algorithm that has survived to the present day.''

The Euclidean Algorithm serves as a subroutine for many tasks such as: finding B\'ezout's coefficients, computing multiplicative inverses and the RSA algorithm, Chinese Remainder Theorem and others.
For more details on these and other applications, see \cite{Knuth97,Knuth98}. Yet, despite its simplicity, the actual runtime complexity cannot be easily inferred from its description due to the recursive nature of the algorithm.

As the operations inside each iteration are ``elementary'', the actual runtime complexity is dominated by the number of (recursive) iterations of the algorithm.  Indeed, various upper bounds for this number were given \cite{Lame1844, Sedgewick83, Knuth97,Knuth98,CLRS09}. The first one, presented in \cite{Lame1844}, ties the number of iteration with the so-called \emph{Fibonacci} numbers. 
Another kind of analysis shows that the larger argument shrinks by a factor of at least $2$ every \emph{two} iterations (see e.g. \cite{Knuth97,CLRS09}). Yet, the proof of this claim requires a somewhat non-trivial case-analysis. 
Note that in the case of univariate polynomials, the analysis is much simpler: the degree of the highest-degree polynomial decreases at each step. A similar easy 
argument was lacking in the integer case since the same phenomenon does 
not occur: it does \textbf{not} hold that the bit-size of the largest integer decreases at each step.

Our new analysis provides a clean, ``one-line'' upper bound which turns out to be optimal.  
The main idea is based on the so-called \emph{Potential Method} (see e.g. \cite{CLRS09}).
More specifically, we define appropriate \emph{potential} functions and show that 
these functions lose a constant fraction of their mass with every recursive iteration of the algorithm. At the same time, the functions are bounded away from zero. Formally, we give a simple proof to the following theorems:

\begin{THEOREM}
\label{THM:main1}
For any $x,y \in \N$, the Euclidean Algorithm performs at most $\log _{1.5} (x+y) + 1$ iterations.
\end{THEOREM}
\noindent Using a slightly more sophisticated potential function yields an even tighter result.   

\begin{THEOREM}
\label{THM:main2}
For any $x,y \in \N$, the Euclidean Algorithm performs at most $\log _{\phi} (\phi x+ y)$ iterations.
\end{THEOREM}
\vspace{-0.1in}
\noindent Here $\phi = \frac{1+\sqrt{5}}{2}  \approxeq 1.61803398875$ denotes the \emph{golden ratio}. 

\medskip

\noindent Our analysis should be contrasted with
the analysis of \cite{Lame1844} which states that if the Euclidean algorithm requires $m$ iterations, then it must be the case that $x \geq F_{m+2}$ and $y \geq F_{m+1}$, where $F_{i}$-s are the Fibonacci numbers: $1,1,2,3,5, \ldots$. 
And, indeed, selecting $x$ and $y$ as consecutive Fibonacci numbers demonstrates the tightness of \textbf{this} analysis. Now for a general pair $x$ and $y$ we have:
\begin{equation*}
\phi x+ y \geq \phi F_{m+2} + F_{m+1} \geq F_{m+3}  \gtrapprox  \phi^{m+1}.
\end{equation*}
Therefore $m \lessapprox \log_{\phi}(\phi x+y)$, implying the optimality of \textbf{our} analysis.


\section{The Algorithm and the Runtime Analysis}

In this section we review the algorithm and give the new, simplified runtime analysis. The description of the algorithm is given below in Algorithm \ref{alg:Euclid}. The algorithm is given non-negative integers, $x$ and $y$ as an input. We assume w.l.o.g that $x > y$.

\medskip

\begin{algorithm}[H] 
\label{alg:Euclid}
    \KwIn{Two non-negative integers, $x > y \in \N$}
   \KwOut{$\gcd(x,y)$}

    \bigskip  
    \lIf{$y=0$}{\Return $x$}{}
    
    \lIf{$y=1$}{\Return $1$}{}
    
    {\Return $\gcd(y, x \bmod y)$}
     \caption{Euclidean Algorithm}
\end{algorithm}

\medskip

For the sake of analysis, let us fix $x > y$ and let $m$ denote the number of the (recursive) iterations of algorithm. 

\begin{definition}
For $1 \leq i \leq m$, let $x_i$ and $y_i$ denote the values of $x$ and $y$ in the $i$-th iteration of the algorithm, respectively. In particular, $x_1 = x, y_1 = y$. 
\end{definition}

\begin{observation}
For $1 \leq i \leq m$: $x_i > y_i \geq 0$.
\end{observation}

\begin{proof}
Follows from the fact that $y_i = x_{i-1} \bmod y_{i-1} = x_{i-1} \bmod x_{i}$, which attains values between $0$ and $x_{i}-1$.
\end{proof}

\begin{definition}
For $1 \leq i \leq m$, we define $s_i \eqdef x_i + y_i$. 
\end{definition}

Next, we will show that $s_i$ is a \emph{potential} function for this algorithm. In particular, the function loses a constant fraction of its mass in every step, yet it is bounded away from zero. The following is immediate from the definition, given the above observation.

\begin{corollary}
\label{cor1}
For $1 \leq i \leq m$: $s_i \geq 1$.
\end{corollary}

The following lemma is the heart of the argument.

\begin{lemma}
\label{lem:main1}
For $1 \leq i \leq m$: $s_{i} \leq 2/3 \cdot s_{i-1}$.
\end{lemma}

\begin{proof}
Recall that $s_{i-1} = x_{i-1} + y_{i-1}$.
Let us divide $x_{i-1}$ by $y_{i-1}$ with a reminder. In particular, 
$$x_{i-1} = q_{i-1} \cdot y_{i-1} + r_{i-1}, \text{ where } 0 \leq r_{i-1} \leq y_{i-1} -1.$$
Observe that $s_{i} = y_{i-1} + r_{i-1}$ and in addition, $q_{i-1} \geq 1$, as $x_{i-1} > y_{i-1}$. We obtain the following:
\begin{equation*}
s_{i-1} = 
(q_{i-1} + 1) y_{i-1} + r_{i-1} \geq 2y_{i-1} + r_{i-1} \geq 2y_{i-1} + r_{i-1} - (y_{i-1} - r_{i-1})/2 = 1.5y_{i-1} + 1.5r_{i-1} = 1.5 s_{i}.
\end{equation*}
The second inequality holds since $r_{i-1} < y_{i-1}$.
\end{proof}

By applying the lemma repeatedly, we obtain:

\begin{corollary}
\label{cor2}
For $1 \leq i \leq m$: $s_i \leq s_1 \cdot (2/3)^{i-1}$.
\end{corollary}

Theorem \ref{THM:main1} follows immediately from the above.

\begin{proof}[Proof of Theorem \ref{THM:main1}]
By Corollaries \ref{cor1} and \ref{cor2}:
$$ 1 \leq s_m \leq s_1 \cdot (2/3)^{m-1} = (x+y) \cdot (2/3)^{m-1}.$$
Therefore, $m \leq \log _{1.5} (x+y) + 1$.
\end{proof}

\subsection{Improved Analysis}

In this section we show that using a slightly more sophisticated potential function yields an even tighter result.  As before, for the sake of analysis, let us fix $x > y$ and let $m$ denote the number of the (recursive) iterations of algorithm. In addition, let $\phi = \frac{1+\sqrt{5}}{2}  \approxeq 1.61803398875$
denote the golden ratio.

\begin{definition}
For $1 \leq i \leq m$, we define $s_i \eqdef x_i + \frac{1}{\phi} \cdot y_i $. 
\end{definition}

As before, for $1 \leq i \leq m$: $s_i \geq 1$. The following lemma mirrors Lemma \ref{lem:main1}.

\begin{lemma}
For $1 \leq i \leq m$: $s_{i} \leq \frac{1}{\phi} \cdot s_{i-1}$.
\end{lemma}

\begin{proof}
Recall that $s_{i-1} = x_{i-1} + \frac{1}{\phi}\cdot y_{i-1}$.
Let us divide $x_{i-1}$ by $y_{i-1}$ with a reminder. In particular, 
$$x_{i-1} = q_{i-1} \cdot y_{i-1} + r_{i-1}, \text{ where } 0 \leq r_{i-1} \leq y_{i-1} -1.$$
Observe that $s_{i} = y_{i-1} + \frac{1}{\phi} \cdot r_{i-1}$ and in addition, $q_{i-1} \geq 1$, as $x_{i-1} > y_{i-1}$. We obtain the following:
\begin{equation*}
s_{i-1} = 
\left( q_{i-1} + \frac{1}{\phi} \right) y_{i-1} + r_{i-1} \geq
\left(1 + \frac{1}{\phi} \right)  y_{i-1} + r_{i-1} =
\phi \cdot y_{i-1} + r_{i-1} = \phi \left(y_{i-1} + \frac{1}{\phi} \cdot r_{i-1}  \right) = \phi \cdot s_{i}.
\end{equation*}
Recall that $1 + \frac{1}{\phi} = \phi.$
\end{proof}

The proof of Theorem \ref{THM:main2} follows by repeating the previous argument.

\section{Conclusion \& Open Questions}

In this short note we add one (more) line of analysis to the well-known Euclidean Algorithm. The new analysis does not require any background and can be taught even in an introductory-level undergraduate class. It would be nice to see if we could replace other kind of analyses that rely on Fibonacci numbers by a potential argument. One such example is the analysis of the height of an AVL-Tree \cite{AVL62,Sedgewick83}.

\bibliographystyle{alpha}
\bibliography{C:/Work/Papers/bibliography}

\end{document}